\documentclass[11pt]{article}
\usepackage[T1]{fontenc}
\usepackage[utf8]{inputenc}
\usepackage{lmodern}
\usepackage{microtype}
\usepackage{fullpage}
\usepackage{authblk}
\usepackage{amsmath,amssymb,amsthm,mathtools}
\usepackage{aliascnt}
\usepackage{bm}
\usepackage{mathpazo}
\usepackage{booktabs}
\usepackage{tabularx}
\usepackage{longtable}
\usepackage{array}
\usepackage{xcolor}
\definecolor{MainBlue}{HTML}{173B57}
\definecolor{AccentTeal}{HTML}{247B7B}
\definecolor{SoftBlue}{HTML}{F1F6FA}
\definecolor{SoftTeal}{HTML}{EFF8F7}
\definecolor{SoftGray}{HTML}{F5F6F7}
\definecolor{RuleGray}{HTML}{D8DEE4}
\definecolor{WarmGold}{HTML}{A66F18}
\usepackage{enumitem}
\usepackage{mdframed}
\usepackage{needspace}
\usepackage[numbers,sort&compress]{natbib}
\usepackage{hyperref}
\usepackage[nameinlink,noabbrev]{cleveref}

\hypersetup{
colorlinks=true,
linkcolor=blue!55!black,
citecolor=blue!55!black,
urlcolor=blue!55!black,
pdfborder={0 0 0}
}
\allowdisplaybreaks
\theoremstyle{plain}
\newtheorem{theorem}{Theorem}[section]
\newaliascnt{proposition}{theorem}
\newtheorem{proposition}[proposition]{Proposition}
\aliascntresetthe{proposition}
\newaliascnt{lemma}{theorem}
\newtheorem{lemma}[lemma]{Lemma}
\aliascntresetthe{lemma}
\newaliascnt{corollary}{theorem}
\newtheorem{corollary}[corollary]{Corollary}
\aliascntresetthe{corollary}
\theoremstyle{definition}
\newaliascnt{definition}{theorem}
\newtheorem{definition}[definition]{Definition}
\aliascntresetthe{definition}
\theoremstyle{remark}
\newaliascnt{remark}{theorem}
\newtheorem{remark}[remark]{Remark}
\aliascntresetthe{remark}

\mdfdefinestyle{theoremframe}{
backgroundcolor=SoftBlue,
linecolor=MainBlue,
linewidth=.65pt,
innerleftmargin=7pt,innerrightmargin=7pt,
innertopmargin=6pt,innerbottommargin=6pt,
needspace=6\baselineskip,
nobreak=true,
skipabove=9pt,skipbelow=9pt
}
\mdfdefinestyle{propositionframe}{
backgroundcolor=SoftTeal,
linecolor=AccentTeal,
linewidth=.6pt,
innerleftmargin=7pt,innerrightmargin=7pt,
innertopmargin=6pt,innerbottommargin=6pt,
needspace=6\baselineskip,
nobreak=true,
skipabove=8pt,skipbelow=8pt
}
\mdfdefinestyle{lemmaframe}{
backgroundcolor=SoftGray,
linecolor=RuleGray!60!black,
linewidth=.5pt,
innerleftmargin=7pt,innerrightmargin=7pt,
innertopmargin=5pt,innerbottommargin=5pt,
needspace=5\baselineskip,
nobreak=true,
skipabove=8pt,skipbelow=8pt
}
\surroundwithmdframed[style=theoremframe]{theorem}
\surroundwithmdframed[style=propositionframe]{proposition}
\surroundwithmdframed[style=propositionframe]{corollary}
\surroundwithmdframed[style=lemmaframe]{lemma}
\surroundwithmdframed[style=theoremframe]{definition}

\crefname{theorem}{theorem}{theorems}
\Crefname{theorem}{Theorem}{Theorems}
\crefname{proposition}{proposition}{propositions}
\Crefname{proposition}{Proposition}{Propositions}
\crefname{lemma}{lemma}{lemmas}
\Crefname{lemma}{Lemma}{Lemmas}
\crefname{corollary}{corollary}{corollaries}
\Crefname{corollary}{Corollary}{Corollaries}
\crefname{definition}{definition}{definitions}
\Crefname{definition}{Definition}{Definitions}
\crefname{remark}{remark}{remarks}
\Crefname{remark}{Remark}{Remarks}
\crefname{equation}{equation}{equations}
\Crefname{equation}{Equation}{Equations}

\newcommand{\cL}{\mathcal L}
\newcommand{\cD}{\mathcal D}
\newcommand{\cN}{\mathcal N}
\newcommand{\cM}{\mathcal M}
\newcommand{\cC}{\mathcal C}
\newcommand{\Tr}{\operatorname{Tr}}
\newcommand{\supp}{\operatorname{supp}}
\newcommand{\id}{\operatorname{id}}

\newcommand{\1}{\mathbb 1}
\newcommand{\e}{\mathrm e}
\newcommand{\ket}[1]{\lvert #1\rangle}
\newcommand{\bra}[1]{\langle #1\rvert}
\newcommand{\norm}[1]{\left\lVert #1\right\rVert}
\newcommand{\sandD}{\widetilde D}
\newcommand{\sandQ}{\widetilde Q}

\newcommand{\Dmax}{D_{\max}}
\newcommand{\parr}{\mathrm{par}}
\newcommand{\ada}{\mathrm{ada}}
\newcommand{\gen}{\mathrm{gen}}

\numberwithin{equation}{section}

\title{\Large\textbf{Quantum Channel Stein Theorem beyond Definite Causal Order}}
\author[1]{Chengkai Zhu}
\author[2,*]{Xin Wang}

\affil[1]{QudeLeap Research, Shanghai 200030, China}
\affil[2]{Thrust of Artificial Intelligence, Information Hub,\par
The Hong Kong University of Science and Technology (Guangzhou),
Guangzhou 511453, China}
\date{September 2026}

\begin{document}
\maketitle

\begingroup
\renewcommand{\thefootnote}{\fnsymbol{footnote}}

\footnotetext[1]{\href{mailto:felixxinwang@hkust-gz.edu.cn}
{felixxinwang@hkust-gz.edu.cn}}

\endgroup

\begin{abstract}
Distinguishability is a central concept in information theory and extends naturally to quantum systems. This work shows that, for any two finite-dimensional memoryless quantum channels, parallel, adaptive and general testing strategies achieve the same Stein exponent at every fixed type-I error tolerance $\varepsilon\in(0,1)$, namely the regularized channel relative entropy. When this rate is finite, any larger type-II exponent forces the probability of correctly accepting the null hypothesis to decay exponentially. Thus, adaptivity provides no asymptotic advantage. The key step is proving the continuity of the regularized sandwiched Rényi channel divergence at order one. We also derive the exact strong-converse exponent for general testers.
\end{abstract}

\setcounter{tocdepth}{1}
\begingroup
\small
\tableofcontents
\endgroup
\clearpage

\section{Introduction}

To distinguish two quantum channels, an experimenter can send a probe through channel uses and discriminate the output quantum states. A parallel strategy applies several channel uses to a joint probe. Repeating a fixed block produces independent output pairs, to which the state Stein lemma applies~\cite{HiaiPetz1991,OgawaNagaoka2000}. Optimizing the block gives the
regularized, reference-assisted channel relative entropy
$D^\infty(\cN\Vert\cM)$ as an achievable rate~\cite{WangWilde2019,FangFawziRennerSutter2020}. For parallel and
adaptive strategies, Fawzi and Fawzi characterized the decay exponent
of correct null acceptance when the Type-II exponent is constrained~\cite[Theorem 5.5]{FawziFawzi2021}. Their formula uses regularized sandwiched R\'enyi divergences of orders greater than one, so its positivity threshold is $\inf_{p>1}\sandD_p^\infty(\cN\Vert\cM)$. Whether this infimum equals $D^\infty(\cN\Vert\cM)$ is the endpoint question identified in~\cite[Remarks 5.3 and 5.6]{FawziFawzi2021}. Related work connects strong converses, smoothing, and continuity while also studying channel discrimination without quantum memory
\cite{FangGourWang2025}. Finite-use adaptive parallelization
\cite{BerghEtAl2024} and equivalence under postselected discrimination
\cite{RegulaLamiWilde2024} concern different operational comparisons;
in the latter task, errors are conditioned on a conclusive outcome.
Earlier amortized-divergence bounds established the strong Stein lemma
for classical--quantum channels \cite{WildeBertaHircheKaur2020}.

Beyond parallel and sequential strategies, indefinite causal order allows channel uses to be embedded in a higher-order process without a predetermined causal order, and can yield a strict advantage in finite-use minimum-error channel discrimination~\cite{OreshkovCostaBrukner2012,BavarescoEtAl2021}. That hierarchy does not determine the asymmetric error exponent. In asymmetric testing, the null channel $\cN$ must be accepted with probability at least
$1-\varepsilon$, while the probability of accepting it when the channel
is $\cM$ is minimized. The existing parallel and adaptive results also leave processes without a definite causal order outside their scope. This motivates a fundamental question:
\begin{center}
\itshape
Can indefinite causal order improve the error exponent in asymmetric channel discrimination?
\end{center}

In this work, we prove that parallel, adaptive, and general testers
have the same fixed-error Stein rate $D^\infty(\cN\Vert\cM)$ for every finite-dimensional channel pair (\Cref{thm:stein}). When this rate is finite, correct null acceptance
decays exponentially at every larger Type-II exponent. Our proof
separately resolves the tester-normalization loss and the R\'enyi
endpoint.

First, we express the weighted binary testing score as the least
normalized positive Choi slack. A fixed-marginal de Finetti reduction
\cite[Corollary 1.1]{NaharEtAl2024} turns this slack into a mixture of
IID channel Choi operators with only polynomial loss. Contracting the
resulting order inequality with a general tester gives the weighted-score
bound used in the converse (\Cref{sec:reduction}). The polynomial factor
has no effect on exponential rates. An explicit conversion of general
tests to parallel tests is given in \Cref{app:parallelization}.

Second, we prove continuity at order one of the regularized sandwiched
R\'enyi channel divergence (\Cref{thm:endpoint}). An elementary supported-inverse
factorization lifts the same positive slack to one amplitude approximation
valid for every input. A single tensor moment estimate then controls
arbitrarily entangled probes. Combining approximations at two costs into
three amplitudes closes the endpoint. Together, the testing reduction
and this endpoint identity prove the main theorem.

With the Stein theorem established, we develop two consequences.
First, sharpening the finite-order testing bound and combining it with
parallel achievability from \cite[Theorem 5.5]{FawziFawzi2021} gives the
exact strong-converse exponent for general testers (\Cref{sec:finite}).
Second, the same positive-slack viewpoint gives a finite-copy connection
between testing and smoothing. The smooth-entropy framework for
finite-blocklength quantum tasks \cite{Tomamichel_2013}, its connection
to hypothesis testing \cite{DattaMosonyiHsiehBrandao2013}, and the
modified state duality of \cite{RegulaLamiDatta2026} provide context.
Parallel testers are exactly dual to
smoothing over all channel Choi operators; general testers are dual
to the positive part of the affine hull of product channel Choi
operators (\Cref{thm:one-shot-duality}). These sets, convex product
mixtures, and mixtures of identical tensor powers all have the
fixed-budget asymptotic limit $D^\infty(\cN\Vert\cM)$
(\Cref{thm:smoothing-aep}).
The finite-copy identities follow from SDP duality, while the AEP uses
the endpoint theorem to make the positive slack exponentially small.


\section{Channel discrimination}
\label{sec:operational}
\label{sec:setup}

All systems are finite-dimensional and all logarithms are natural. Write
$\cL(A)$ and $\cD(A)$ for the operators and states on $A$, and $d_A=\dim A$.
Channels are completely positive, trace-preserving (CPTP) maps.
For a Hermitian $X$, $X_+$ is its positive part; $\norm{X}_\infty$ denotes
the operator norm. We write $\psi=\ket\psi\bra\psi$ for a pure state.
We use the unnormalized
Choi operator
\begin{equation}
J^{\cC}_{AB}=(\id_A\otimes\cC)(\ket\Gamma\bra\Gamma),\qquad
\ket\Gamma=\sum_{i=1}^{d_A}\ket i_A\ket i_{A'},\qquad
\Tr_B J^{\cC}=\1_A.
\label{eq:choi}
\end{equation}
Here $A'\simeq A$, and transposes refer to this fixed basis. Tensor factors
are grouped as $A^nB^n$ unless chronological slot order is displayed.

For states, the hypothesis-testing relative entropy is
\begin{equation}
D_H^\varepsilon(\rho\Vert\sigma)
=-\log\inf_{\substack{0\le\Lambda\le\1\\
    \Tr[(\1-\Lambda)\rho]\le\varepsilon}}
\Tr(\Lambda\sigma),\qquad 0\le\varepsilon<1.
\label{eq:state-testing}
\end{equation}
An $n$-slot binary tester is a pair of positive operators $T_1,T_2$ on
$A^nB^n$. Outcome 1 accepts $\cN$; outcome 2 accepts $\cM$.
With the unnormalized Choi convention \eqref{eq:choi}, the two error
probabilities are
\begin{equation}
\Tr[T_2(J^{\cN})^{\otimes n}]
\quad\hbox{(Type I)},\qquad
\Tr[T_1(J^{\cM})^{\otimes n}]
\quad\hbox{(Type II)}.
\label{eq:tester-errors}
\end{equation}
The admissible normalization of $T_1+T_2$ specifies the strategy
\cite{ChiribellaEbler2016}. We use the following explicit definitions.

\begin{definition}[Parallel strategy]
For $0\le\varepsilon<1$, set
\begin{equation}
\begin{aligned}
    \beta_\varepsilon^{\parr}(\cN^{\otimes n},\cM^{\otimes n})
    =\inf\;&\Tr[T_1(J^{\cM})^{\otimes n}]\\
    \text{subject to }&T_1,T_2\ge0,\quad
    \Tr[T_2(J^{\cN})^{\otimes n}]\le\varepsilon,\\
    &T_1+T_2=\sigma_{A^n}\otimes\1_{B^n},\quad
    \sigma_{A^n}\ge0,\quad\Tr\sigma_{A^n}=1.
\end{aligned}
\label{sdp:parallel}
\end{equation}
The state $\sigma$ is the tester normalizer; its transpose is the marginal
of the physical channel input. Indeed, the probe
$(\sqrt\sigma\otimes\1)\ket\Gamma^{\otimes n}$ realizes this convention.
\end{definition}

\begin{definition}[Adaptive strategy]
For $0\le\varepsilon<1$, set
\begin{equation}
\begin{aligned}
    \beta_\varepsilon^{\ada}(\cN^{\otimes n},\cM^{\otimes n})
    =\inf\;&\Tr[T_1(J^{\cM})^{\otimes n}]\\
    \text{subject to }&T_1,T_2\ge0,\quad
    \Tr[T_2(J^{\cN})^{\otimes n}]\le\varepsilon,\\
    &T_1+T_2=R^{(n)}\otimes\1_{B_n},\quad
    R^{(k)}\ge0\quad(1\le k\le n),\\
    &\Tr_{A_k}R^{(k)}=R^{(k-1)}\otimes\1_{B_{k-1}}
    \quad(2\le k\le n),\\
    &\Tr R^{(1)}=1.
\end{aligned}
\label{eq:comb}
\end{equation}
The slots are used in order from 1 to $n$, and the intermediate normalizers act on
\[
R^{(k)}\in\cL(A_1B_1\cdots A_{k-1}B_{k-1}A_k).
\]
These are deterministic tester (co-comb) constraints
\cite{ChiribellaDArianoPerinotti2009}; grouped tensor
order $A^nB^n$ is related to chronological order by a fixed permutation.
They describe sequential experiments with arbitrary intermediate channels
and retained quantum memory. At $n=1$ they reduce to the parallel constraint.
\end{definition}

\begin{definition}[General strategy]
For $0\le\varepsilon<1$, set
\begin{equation}
\begin{aligned}
    \beta_\varepsilon^{\gen}(\cN^{\otimes n},\cM^{\otimes n})
    =\inf\;&\Tr[T_1(J^{\cM})^{\otimes n}]\\
    \text{subject to }&T_1,T_2\ge0,\quad
    \Tr[T_2(J^{\cN})^{\otimes n}]\le\varepsilon,\\
    &\Tr[(T_1+T_2)(J^{\cC_1}\otimes\cdots\otimes J^{\cC_n})]=1\\
    &\hspace{1em}\text{for every collection of CPTP maps }
    \cC_1,\ldots,\cC_n:A\to B.
\end{aligned}
\label{eq:normalization}
\end{equation}
This is the positive process-tester model, which permits indefinite causal
order. Normalization is imposed on arbitrary products of local channels,
not on arbitrary joint channels from $A^n$ to $B^n$. We optimize over this
entire mathematical class; no claim that every feasible tester has a
particular circuit realization is needed. The converse therefore also
applies to any physically realizable subclass contained in it.
\end{definition}

For $\mathsf S\in\{\parr,\ada,\gen\}$, define the hypothesis-testing channel relative entropy by
\begin{equation}
D_H^{\varepsilon,\mathsf S}(\cN^{\otimes n}\Vert\cM^{\otimes n})
=-\log\beta_\varepsilon^{\mathsf S}
(\cN^{\otimes n},\cM^{\otimes n}),
\label{eq:roc}
\end{equation}
with $-\log0=+\infty$. Each normalizer gives a total probability of one under
both hypotheses, so the Type-I condition is equivalent to null acceptance
at least $1-\varepsilon$.

\paragraph{Inclusions of the tester classes.}
A parallel normalizer in \eqref{sdp:parallel} satisfies the adaptive
constraints \eqref{eq:comb} with
\[
\Xi^{(k)}=\Tr_{A_{k+1}\cdots A_n}\sigma_{A^n}
\otimes\1_{B_1\cdots B_{k-1}}.
\]
This construction permits entangled $\sigma_{A^n}$.
For an adaptive normalizer, contract the last Choi operator using
$\Tr_{B_n}J^{\cC_n}=\1_{A_n}$, then apply the recursion in
\eqref{eq:comb}. Repeating the contraction leaves $\Tr\Xi^{(1)}=1$.
Thus every adaptive tester obeys \eqref{eq:normalization}, proving
$\parr\subseteq\ada\subseteq\gen$.

The general tester set is compact. Indeed, inserting the completely
depolarizing channel in every slot gives
$\Tr(T_1+T_2)=d_B^n$, which bounds both positive effects.
The normalization conditions are closed linear constraints.
Parallel testers form a compact set as well, by
\eqref{sdp:parallel}. Their fixed-error minima are therefore attained. In proofs we write $W=T_1+T_2$ for the deterministic normalizer.
The \emph{acceptance pair} $(a,b)$ of a tester $(T_1,T_2)$ consists of
its probabilities of accepting the null hypothesis $\cN$ when the
inserted channel is $\cN$ and $\cM$, respectively:
\begin{equation}
a=\Tr[T_1(J^{\cN})^{\otimes n}],\qquad
b=\Tr[T_1(J^{\cM})^{\otimes n}].
\label{eq:probabilities}
\end{equation}
Thus, $1-a$ is the Type-I error and $b$ is the Type-II error.

\paragraph{The achievable rate.} Repeating a fixed block probe turns channel discrimination into state discrimination. Its rate is the relative entropy of the two output
states. For states with $\supp\rho\subseteq\supp\sigma$, define
\begin{equation}
D(\rho\Vert\sigma)=\Tr\rho(\log\rho-\log\sigma),
\label{eq:state-relative-entropy}
\end{equation}
and set it to $+\infty$ otherwise. Optimizing this rate over the probe
gives the reference-assisted channel divergence. With $\mathbb D=D$, set
\begin{equation}
\mathbb D(\cN\Vert\cM)
=\sup_{\substack{\psi_{RA}\ \mathrm{pure}\\R\simeq A}}
\mathbb D\bigl((\id_R\otimes\cN)(\psi)\Vert
(\id_R\otimes\cM)(\psi)\bigr).
\label{eq:channel-divergence}
\end{equation}
Purification, Schmidt compression, and data processing justify this
reference dimension. Product inputs and additivity of state relative entropy
make the block sequence superadditive. Consequently, Fekete's lemma gives
\begin{equation}
\mathbb D^\infty(\cN\Vert\cM)
=\lim_{k\to\infty}\frac1k\mathbb D(\cN^{\otimes k}\Vert\cM^{\otimes k})
=\sup_{k\ge1}\frac1k\mathbb D(\cN^{\otimes k}\Vert\cM^{\otimes k}).
\label{eq:regularization}
\end{equation}
In particular, the same limit holds along the multiples of any fixed
blocklength. For positive operators, the max-relative entropy is
$\Dmax(\rho\Vert\sigma)=\log\inf\{c\ge0:\rho\le c\sigma\}$,
with value $+\infty$ when no finite $c$ exists. Its channel version is
\begin{equation}
\Dmax(\cN\Vert\cM)=\Dmax(J^{\cN}\Vert J^{\cM})
=\log\inf\{c\ge0:J^{\cN}\le cJ^{\cM}\}.
\label{eq:rates}
\end{equation}
It is finite exactly when $\supp J^{\cN}\subseteq\supp J^{\cM}$. Choi order is equivalent to complete-positive order, so this value bounds the max-relative entropy of the outputs for every input and reference.
The maximally entangled probe attains it because its outputs are
$J^{\cN}/d_A$ and $J^{\cM}/d_A$. Thus, it also agrees with the reference-assisted output definition of channel max-relative entropy.

\begin{theorem}[Channel Stein theorem and exponential strong converse]
\label{thm:stein}
For every pair of finite-dimensional CPTP maps $\cN,\cM:A\to B$, every
fixed $\varepsilon\in(0,1)$, and every
$\mathsf S\in\{\parr,\ada,\gen\}$,
\begin{equation}
\lim_{n\to\infty}\frac1nD_H^{\varepsilon,\mathsf S}(\cN^{\otimes n}\Vert\cM^{\otimes n})
=D^\infty(\cN\Vert\cM).
\label{eq:stein}
\end{equation}
Equality holds also when the rate is infinite. If the rate is finite,
then for every $r>D^\infty(\cN\Vert\cM)$ there exist $c>0$ and $n_0$,
depending only on the channels and $r$, such that every general tester
with $b\le\e^{-nr}$ satisfies $a\le\e^{-cn}$ for all $n\ge n_0$.
\end{theorem}

The restriction $\varepsilon>0$ is necessary in general. For example,
for replacer channels with distinct full-rank commuting outputs $\rho$
and $\sigma$, zero Type-I error forces acceptance on the full support
of $\rho^{\otimes n}$, giving $D_H^{0,\parr}=0$ despite
$D(\rho\Vert\sigma)>0$.

The theorem identifies both the best fixed-error rate and the behavior
above it: demanding a larger Type-II exponent forces the probability
of correctly accepting the null channel to vanish exponentially, even
for general processes. We first prove parallel achievability and reduce
the general converse to a scalar testing bound in \Cref{sec:reduction}.
\Cref{sec:amplitudes} identifies its threshold and completes the proof.

\section{The exact one-shot duality}
\label{sec:one-shot}
We identify the smoothing set dual to each tester normalization. This section is independent of the asymptotic converse: its duality will later give the lower bound in the smoothing AEP. Throughout this section, use the Choi abbreviations
$N_n=(J^{\cN})^{\otimes n}$ and $M_n=(J^{\cM})^{\otimes n}$.

Let
\[
\mathcal P_n=
\{J^{\cC_1}\otimes\cdots\otimes J^{\cC_n}:
\cC_1,\ldots,\cC_n:A\to B\ \mathrm{CPTP}\}.
\]
We distinguish the following sets of Choi operators:
\begin{equation}
\begin{aligned}
    \mathcal K_{\mathrm{all},n}
    &=\{Q\ge0:\Tr_{B^n}Q=\1_{A^n}\},\\
    \mathcal K_{\mathrm{aff},n}
    &=\operatorname{aff}_{\mathbb R}(\mathcal P_n)\cap\{Q:Q\ge0\},\\
    \mathcal K_{\mathrm{prod},n}
    &=\operatorname{conv}(\mathcal P_n),\\
    \mathcal K_{\mathrm{iid},n}
    &=\operatorname{conv}\{(J^{\cC})^{\otimes n}:
    \cC:A\to B\ \mathrm{CPTP}\}.
\end{aligned}
\label{eq:smoothing-sets}
\end{equation}
Here $\operatorname{aff}_{\mathbb R}$ denotes the real affine hull.
These nonempty convex compact sets satisfy
\begin{equation}
\mathcal K_{\mathrm{iid},n}\subseteq
\mathcal K_{\mathrm{prod},n}\subseteq
\mathcal K_{\mathrm{aff},n}\subseteq
\mathcal K_{\mathrm{all},n}.
\label{eq:smoothing-inclusions}
\end{equation}
Indeed, every affine combination in \eqref{eq:smoothing-sets} has
partial trace $\1_{A^n}$. Its positive elements therefore lie in
the compact set $\mathcal K_{\mathrm{all},n}$, and the affine hull
is closed in finite dimensions. The other two sets are convex hulls
of compact sets in a finite-dimensional space, and hence are compact.
The affine constraint is a linear normalization constraint; it does
not impose separability across slots.
The IID set consists of mixtures of identical tensor powers; the
mixture itself need not be a product operator.

For $\bullet\in\{\mathrm{all},\mathrm{aff},\mathrm{prod},\mathrm{iid}\}$
and $0\le\delta<1$, define the modified smooth max-relative entropy by
\begin{equation}
\sandD_{\max,\bullet}^{\delta}
(\cN^{\otimes n}\Vert\cM^{\otimes n})
=\log\inf\{\lambda\ge0:
N_n\le\lambda M_n+\delta Q,\quad Q\in\mathcal K_{\bullet,n}\},
\label{eq:modified-smoothing}
\end{equation}
with value $+\infty$ if no finite $\lambda$ is feasible.
The budget $\delta$ multiplies a normalized Choi operator in the
specified set and therefore measures an additive positive correction.
The definition does not require a nearby CPTP approximation or impose
a distance constraint on separately smoothed output states.
This convention extends the modified smoothing
used in the state duality of \cite{RegulaLamiDatta2026}.
Taking traces shows that every feasible $\lambda$ obeys
$\lambda\ge1-\delta$. Under Choi support inclusion, direct domination
also gives
\begin{equation}
1-\delta\le
\exp\!\left[\sandD_{\max,\bullet}^{\delta}
(\cN^{\otimes n}\Vert\cM^{\otimes n})\right]
\le\e^{n\Dmax(\cN\Vert\cM)}.
\label{eq:smoothing-compact-bounds}
\end{equation}
The infimum is then attained: restrict $\lambda$ to this compact
interval and use compactness of the Choi set. In particular, the
entropy may be negative. The set inclusions imply, at the same
arguments and budget,
\begin{equation}
\sandD_{\max,\mathrm{all}}^\delta
\le\sandD_{\max,\mathrm{aff}}^\delta
\le\sandD_{\max,\mathrm{prod}}^\delta
\le\sandD_{\max,\mathrm{iid}}^\delta.
\label{eq:smoothing-order-chain}
\end{equation}

Similar to the state case, we have the following duality between the hypothesis testing relative entropy and the modified smooth max-relative entropy.

\begin{theorem}[Exact one-shot duality]
\label{thm:one-shot-duality}
Let $\cN,\cM:A\to B$ be finite-dimensional CPTP maps satisfying
$\supp J^{\cN}\subseteq\supp J^{\cM}$.
For every $n\ge1$ and $q\in(0,1)$,
\begin{align}
    &D_H^{1-q,\parr}(\cN^{\otimes n}\Vert\cM^{\otimes n})=\inf_{0\le\delta<q}
    \left\{\sandD_{\max,\mathrm{all}}^\delta
    (\cN^{\otimes n}\Vert\cM^{\otimes n})-\log(q-\delta)\right\},
    \label{eq:one-shot-par}\\
    &D_H^{1-q,\gen}(\cN^{\otimes n}\Vert\cM^{\otimes n})=\inf_{0\le\delta<q}
    \left\{\sandD_{\max,\mathrm{aff}}^\delta
    (\cN^{\otimes n}\Vert\cM^{\otimes n})-\log(q-\delta)\right\}.
    \label{eq:one-shot-gen}
\end{align}
Here $q$ is the required null acceptance and $\delta$ is the
smoothing budget.
\end{theorem}

\begin{proof}
Write $\beta_n^{\mathsf S}(q)=
\beta_{1-q}^{\mathsf S}(\cN^{\otimes n},\cM^{\otimes n})$ within this proof. For parallel testers, $\sigma\ge0$ follows from
$T_1+T_2=\sigma\otimes\1\ge0$ and need not be imposed separately.
Treating $\sigma$ as Hermitian, the SDP dual of \eqref{sdp:parallel} is
\begin{equation}
 \beta_n^{\parr}(q)=
 \sup_{\substack{u\ge0,\ Z\ge0,\ v\in\mathbb R\\
             uN_n\le M_n+Z,\ \Tr_{B^n}Z=v\1_{A^n}}}(uq-v).
 \label{eq:dual-par-sdp}
\end{equation}
Indeed, use multipliers $u\ge0$ for the null-success constraint,
$Z=Z^\dagger$ for $T_1+T_2=\sigma\otimes\1$, and $v$ for
$\Tr\sigma=1$. Minimizing over $T_1,T_2\ge0$ gives
$M_n+Z-uN_n\ge0$ and $Z\ge0$; minimizing over Hermitian $\sigma$
gives the partial-trace equality. For $q<\xi<1$, the point
$\sigma=\1/d_A^n$, $T_1=\xi\1/d_A^n$, $T_2=(1-\xi)\1/d_A^n$
is strictly feasible, so strong duality holds. Positivity forces
$v\ge0$. If $v>0$, then $Z=vQ$ with
$Q\in\mathcal K_{\mathrm{all},n}$; if $v=0$, then $Z=0$.
No completion of a subnormalized slack is required.

\smallskip
For general testers, let $\Delta$ be the completely depolarizing
channel, put $P_0=(J^\Delta)^{\otimes n}$, and define
\[
\mathcal V_n=\operatorname{span}_{\mathbb R}
\{P-P_0:P\in\mathcal P_n\}.
\]
Choose a Hermitian real basis $F_1,\ldots,F_m$ of $\mathcal V_n$.
General normalization is equivalent to
$\Tr WP_0=1$ and $\Tr WF_j=0$ for all $j$, where $W=T_1+T_2$.
The Lagrangian multipliers for these equalities give
$Z=vP_0+\sum_jz_jF_j$. Minimization over positive $T_1,T_2$
therefore yields
\begin{equation}
    \beta_n^{\gen}(q)=
    \sup_{\substack{u\ge0,\ Z\ge0,\ v\in\mathbb R\\
            uN_n\le M_n+Z,\ Z\in vP_0+\mathcal V_n}}(uq-v).
    \label{eq:dual-gen-sdp}
\end{equation}
The same strictly positive $T_1,T_2$ satisfy every general
normalization equality, so strong duality again holds.
Each $F_j$ has zero partial trace on $B^n$. Consequently,
$\Tr_{B^n}Z=v\1$, which implies $v\ge0$. If $v>0$, then
\[
Z/v\in(P_0+\mathcal V_n)\cap\{Q:Q\ge0\}
=\mathcal K_{\mathrm{aff},n}.
\]
If $v=0$, then $Z=0$, and any element of this set represents
the zero term. This identifies the affine dual set without
assuming a convex product decomposition.

\smallskip
In either case we have obtained
\[
\beta_n^{\mathsf S}(q)=
\sup_{\substack{u,v\ge0,\ Q\in\mathcal K_{\bullet,n}\\
        uN_n\le M_n+vQ}}(uq-v),
\]
with $(\mathsf S,\bullet)=(\parr,\mathrm{all})$ or
$(\gen,\mathrm{aff})$. Choi domination implies
\[
\beta_n^{\mathsf S}(q)\ge
q\,\e^{-n\Dmax(\cN\Vert\cM)}>0.
\]
Thus candidates with nonpositive objective may be discarded.
Every remaining candidate has $u>0$, and the invertible change
of variables
\[
\lambda=1/u,\qquad \delta=v/u
\]
gives $\lambda>0$, $0\le\delta<q$, and
$N_n\le\lambda M_n+\delta Q$, with objective $(q-\delta)/\lambda$.
Conversely, every such feasible smoothing candidate gives a dual
candidate by the inverse substitution. Optimizing $\lambda$ at each
$\delta$ and taking the negative logarithm proves both identities.
The endpoint $\delta=0$ is included throughout; no inverse or
logarithm of zero is taken.
\end{proof}

\section{From testing scores to positive Choi bounds}
\label{sec:reduction}

\paragraph{Parallel achievability.}
The lower bound in \Cref{thm:stein} needs only parallel experiments.
Fix a blocklength $k$ and a pure reference-assisted
probe for $\cN^{\otimes k}$ and $\cM^{\otimes k}$. Repeating this probe
$\lfloor n/k\rfloor$ times and ignoring the remaining uses gives
independent output pairs. The state Stein lemma
\cite{OgawaNagaoka2000} achieves their relative entropy per block
at every fixed $\varepsilon\in(0,1)$. Since
$\lfloor n/k\rfloor/n\to1/k$, optimizing the probe and then $k$,
and using $\parr\subseteq\ada\subseteq\gen$, gives
\begin{equation}
\liminf_{n\to\infty}\frac1n
D_H^{\varepsilon,\mathsf S}(\cN^{\otimes n}\Vert\cM^{\otimes n})
\ge D^\infty(\cN\Vert\cM),\qquad
\mathsf S\in\{\parr,\ada,\gen\},\quad\varepsilon\in(0,1).
\label{eq:stein-lower-bound}
\end{equation}
If Choi support inclusion fails, take
$\rho=J^{\cN}/d_A$, $\sigma=J^{\cM}/d_A$, and let $P_0$ project onto
$\ker\sigma$. Then $q_0=\Tr(P_0\rho)>0$. Independent maximally entangled
probes and acceptance whenever at least one output yields $P_0$ have
$b=0$ and $a=1-(1-q_0)^n\to1$. Thus the operational rate is $+\infty$ for
each fixed $\varepsilon>0$. The one-copy Choi probe also gives
$D(\cN\Vert\cM)=+\infty$, settling the infinite-rate case.

\paragraph{A weighted score for the converse.}
For the upper bound, we must control the null acceptance $a$ when the
alternative acceptance $b$ is exponentially small. Given a penalty
$\lambda>0$, the score $a-\lambda b$ captures this tradeoff. We begin
with parallel tests, whose normalization has a simple SDP description,
and denote their optimal score by
\[
 \Delta_n(\lambda):=\max_{\text{parallel testers}}(a-\lambda b).
\]
Our aim is to transfer this parallel quantity to general testers. We
will prove that every general tester satisfies
\[
 a\le\lambda b+g_n\Delta_n(\lambda),\qquad \log g_n=O(\log n),
\]
with the explicit factor $g_n$ given below. To see why this is useful,
set $\lambda=\e^{nt}$ when $b\le\e^{-nr}$ and $t<r$. The first term is
at most $\e^{-n(r-t)}$; the converse will follow once we show that
$\Delta_n(\e^{nt})$ also decays exponentially for every $t>D^\infty$.

Put $N_n=(J^{\cN})^{\otimes n}$ and
$M_n=(J^{\cM})^{\otimes n}$. In Choi form, the optimal parallel score is
\begin{equation}
 \Delta_n(\lambda)=\max_{\substack{\Tr\sigma=1\,,\ \sigma=\sigma^\dagger\\
                 0\le T\le\sigma\otimes\1_{B^n}}}
       \Tr T(N_n-\lambda M_n),\qquad \lambda>0.
 \label{eq:hockey-primal}
\end{equation}
The constraint $\sigma\ge0$ is automatic: $\sigma\otimes\1\ge T\ge0$.
Thus these are exactly the parallel tests. Using the canonical probe
from \eqref{sdp:parallel}, equivalently
\begin{equation}
 \Delta_n(\lambda)=\sup_\psi\Tr\bigl[(\id\otimes\cN^{\otimes n})(\psi)
             -\lambda(\id\otimes\cM^{\otimes n})(\psi)\bigr]_+.
 \label{eq:hockey}
\end{equation}
In particular, $0\le\Delta_n(\lambda)\le1$. This is an auxiliary
slack budget, rather than an additional channel entropy. The next lemma
connects it directly to the modified smooth $D_{\max}$ in
\Cref{sec:one-shot}.

The advantage of this score is that its dual has a direct operator
interpretation: $\Delta_n(\lambda)$ is the smallest normalized positive
correction needed to dominate $N_n$ by $\lambda M_n$.

\begin{lemma}[Normalized testing slack]
\label{lem:testing-slack}
For every $n\ge1$ and $\lambda>0$,
\begin{equation}
 \Delta_n(\lambda)=\min\{h:\ Z\ge0,\ N_n\le\lambda M_n+Z,
                         \ \Tr_{B^n}Z=h\1_{A^n}\}.
 \label{eq:smoothing}
\end{equation}
Equivalently, $\Delta_n(\lambda)$ is the least $h\ge0$ for which
$N_n\le\lambda M_n+hQ$ with $Q$ a channel Choi operator.
An optimal slack can be chosen permutation invariant.
\end{lemma}
\begin{proof}
For multipliers $Z\ge0$ and $h\in\mathbb R$, the Lagrangian of
\eqref{eq:hockey-primal} is
\[
 h+\Tr T(N_n-\lambda M_n-Z)
       +\Tr\sigma(\Tr_{B^n}Z-h\1).
\]
Its supremum over $T\ge0$ and Hermitian $\sigma$ is finite exactly under
the constraints in \eqref{eq:smoothing}. Slater's condition holds at
$\sigma=\1/d_A^n$ and $T=\1/(2d_A^n)$. Moreover, $Z=N_n$, $h=1$ is
dual feasible, and restricting $0\le h\le1$ gives a compact feasible
set because $\Tr Z=d_A^nh$. This proves equality and attainment.
Positivity forces $h\ge0$; if $h=0$, then $Z=0$. Otherwise write
$Z=hQ$. Finally, joint permutation twirling preserves every constraint
and the value of $h$.
\end{proof}

For $0\le\delta<1$ and $\lambda>0$, the lemma and compactness give
\begin{equation}
 \Delta_n(\lambda)\le\delta
 \quad\Longleftrightarrow\quad
 \sandD_{\max,\mathrm{all}}^\delta
 (\cN^{\otimes n}\Vert\cM^{\otimes n})\le\log\lambda.
 \label{eq:slack-dmax-inverse}
\end{equation}
Indeed, a smaller slack $hQ$ can be increased to $\delta Q$ because
$Q\ge0$. Conversely, a feasible modified smoother is a slack of budget
$\delta$. If the smoothing infimum is finite, it is attained even without
support inclusion: restrict $\lambda$ to any finite feasible sublevel
and use compactness of $\mathcal K_{\mathrm{all},n}$.
Thus $\Delta_n$ is simply the inverse budget profile of the same
modified smooth max-relative entropy.

\paragraph{Transferring the score to general testers.}
The normalized slack $Q$ may be the Choi operator of a joint channel
across all $n$ uses. General testers are normalized on products of local
channels, so we need a product-channel bound on $Q$. The permutation
symmetry supplied by \Cref{lem:testing-slack} allows a fixed-marginal
de Finetti theorem to provide such a bound with polynomial loss. Set
\begin{equation}
 g_n=\binom{n+d_A^2d_B^2-1}{n}.
 \label{eq:postselection-factor}
\end{equation}
At fixed $d_A,d_B$, this satisfies $\log g_n=O(\log n)$.
Let $Q\ge0$ obey $\Tr_{B^n}Q=\1_{A^n}$ and be invariant under
simultaneous permutations of its input--output pairs. The state
\[
 \rho_{A^nB^n}=Q/d_A^n,\qquad
 \Tr_{B^n}\rho=(\1_A/d_A)^{\otimes n}
\]
satisfies the hypotheses of the fixed-marginal de Finetti theorem
\cite[Corollary 1.1]{NaharEtAl2024}. It gives a probability measure
on states $\omega_{AB}$ with $\Tr_B\omega=\1_A/d_A$ such that
\[
 \frac{Q}{d_A^n}\le g_n\int\omega_{AB}^{\otimes n}\,\mathrm d\mu(\omega).
\]
Each $d_A\omega$ is the unnormalized Choi operator of a CPTP map.
Multiplying by $d_A^n$ and pushing the measure forward to these maps gives
\begin{equation}
 Q\le g_n\int(d_A\omega_{AB})^{\otimes n}\,\mathrm d\mu(\omega)
   =g_n\int(J^{\cC})^{\otimes n}\,\mathrm d\mu(\cC).
 \label{eq:definetti}
\end{equation}
Thus no additional dimension factor appears. Permutation invariance,
not support in the symmetric subspace, is the required hypothesis.
Applying this to an optimal normalized slack in \Cref{lem:testing-slack}
yields
\begin{equation}
 N_n\le\lambda M_n+g_n\Delta_n(\lambda)Q_{n,\lambda},\qquad
 Q_{n,\lambda}=\int(J^{\cC})^{\otimes n}\,\mathrm d\mu_{n,\lambda}(\cC).
 \label{eq:slack-lift}
\end{equation}
The mixture depends on the channel pair, $n$, and $\lambda$, but not
on the tester. If $\Delta_n(\lambda)=0$, any such mixture may be used.
For a general acceptance
effect $0\le T_1\le W$, normalization gives $\Tr T_1Q_{n,\lambda}\le1$.
Hence every general tester with acceptance pair $(a,b)$ obeys
\begin{equation}
 a\le\lambda b+g_n\Delta_n(\lambda)\qquad(\lambda>0).
 \label{eq:general-score-bound}
\end{equation}
The operator inequality \eqref{eq:slack-lift} gives the converse through
\eqref{eq:general-score-bound} and will later supply the smoother in the AEP. An explicit parallel test with acceptance pair
$(a/g_n,b/g_n)$ is also available; the stronger finite-use construction
is proved separately in \Cref{app:parallelization}.

\subsection{Binary R\'enyi bounds}

A fixed-order binary bound controls $\Delta_n$ above its R\'enyi threshold.
For the sandwiched R\'enyi divergence
\cite{WildeWinterYang2014,MuellerLennertEtAl2013}, we use the following convention.
For $p>1$ and supported state pairs, define
\begin{align}
\sandD_p(\rho\Vert\sigma)&=\frac{1}{p-1}\log\sandQ_p(\rho\Vert\sigma),
&\sandQ_p(\gamma\Vert\sigma)
&=\Tr\bigl(\sigma^{\frac{1-p}{2p}}\gamma
\sigma^{\frac{1-p}{2p}}\bigr)^p\quad(p>1).
\label{eq:divergences}
\end{align}
We use the same R\'enyi trace expression $\sandQ_p$ for positive
$\gamma$ of arbitrary trace, with $\sandQ_p(0\Vert\sigma)=0$.
Inverse powers are taken on $\supp\sigma$; divergences are
$+\infty$ if support inclusion fails.

Use \eqref{eq:channel-divergence} and \eqref{eq:regularization}
also for $\mathbb D=\sandD_p$. The same purification and superadditivity
arguments apply at every fixed $p>1$.
Under Choi support inclusion, data processing for the sandwiched
divergence \cite{Beigi2013,FrankLieb2013} and monotonicity in the order
\cite{MuellerLennertEtAl2013} imply
\Needspace{5\baselineskip}
\begin{align}
0\le D^\infty(\cN\Vert\cM)
&\le\inf_{\alpha>1}\sandD_\alpha^\infty(\cN\Vert\cM)
\le\sandD_p^\infty(\cN\Vert\cM)\nonumber\\
&\le\Dmax(\cN\Vert\cM)<\infty\qquad(p>1).
\label{eq:rate-order}
\end{align}
Monotonicity in the order identifies the displayed infimum with
$\lim_{p\downarrow1}\sandD_p^\infty(\cN\Vert\cM)$, without yet
identifying its value.

Unless explicitly stated otherwise, assume Choi support inclusion
from here to the end of \cref{sec:finite}.
For a parallel acceptance pair, $a>0$ then implies $b>0$. The case
$a=b=0$ contributes zero to the binary moment. For $p>1$, binary
data processing bounds the moment of every parallel test by
\begin{equation}
 a^pb^{1-p}\le
 \exp\bigl[(p-1)\sandD_p(\cN^{\otimes n}\Vert\cM^{\otimes n})\bigr].
 \label{eq:binary-data-processing}
\end{equation}
When $a>\lambda b$, the positive score satisfies
$(a-\lambda b)\lambda^{p-1}\le a^pb^{1-p}$. Maximizing over tests gives
\begin{equation}
 \Delta_n(\lambda)\le\lambda^{1-p}
 \exp\bigl[(p-1)\sandD_p(\cN^{\otimes n}\Vert\cM^{\otimes n})\bigr].
 \label{eq:hockey-renyi}
\end{equation}
Thus $\Delta_n(\e^{nt})$ decays exponentially whenever
$t>\inf_{p>1}\sandD_p^\infty(\cN\Vert\cM)$. The remaining analytic
step is to identify this threshold with the achievable rate
$D^\infty(\cN\Vert\cM)$.

\section{The regularized endpoint and the Stein converse}
\label{sec:amplitudes}

The preceding reduction leaves one question: does the R\'enyi threshold
coincide with the parallel achievable rate? The following identity
answers this question and is the analytic step needed for the converse.

\begin{theorem}[Regularized R\'enyi endpoint]
\label{thm:endpoint}
For finite-dimensional CPTP maps $\cN,\cM:A\to B$ satisfying
$\supp J^{\cN}\subseteq\supp J^{\cM}$,
\begin{equation}
\lim_{p\downarrow1}\sandD_p^\infty(\cN\Vert\cM)
=D^\infty(\cN\Vert\cM).
\label{eq:endpoint}
\end{equation}
\end{theorem}

The lower inequality follows from \eqref{eq:rate-order}; the content is
continuity after input optimization and regularization. If Choi support
inclusion fails, both sides are infinite by the Choi probe. The proof
below proceeds from the testing slack and binary divergence bounds:
we first construct common amplitude approximations, then tensorize them,
and finally compare two approximation rates.

Assume $C:=\Dmax(\cN\Vert\cM)<\infty$, and write
\[
 D_*:=D^\infty(\cN\Vert\cM),\qquad
 R_*:=\inf_{p>1}\sandD_p^\infty(\cN\Vert\cM).
\]
Then $D_*\le R_*\le C$. We show that no strict gap is possible.
The common approximations will be constructed directly from the
optimizer defining $\Delta_k$.

\subsection{A positive slack gives one common amplitude}
\label{sec:common-approximation}
To use the slack across repeated blocks, we need an approximation that
controls all input probes at once. Working with amplitudes of Choi
operators provides this uniform control and makes tensor products
multiplicative. For a matrix $X$ on $A^kB^k$, define the amplitude norm
\begin{equation}
 \nu_k(X):=\bigl\|\Tr_{B^k}XX^\dagger\bigr\|_\infty^{1/2}
 =\sup_{\sigma\in\cD(A^k)}
       \bigl\|(\sqrt\sigma\otimes\1)X\bigr\|_2.
 \label{eq:amplitude-identities}
\end{equation}
The second expression is a supremum of Hilbert--Schmidt seminorms,
so $\nu_k$ satisfies the triangle inequality. Partial trace gives
\[
 \nu_k(\sqrt{N_k})=1,\qquad
 \nu_{k+\ell}(X\otimes Y)=\nu_k(X)\nu_\ell(Y).
\]
If the columns of $X$ are vectorized Kraus operators, $\nu_k(X)$ is
exactly the operator norm of their common-environment dilation.
Thus the norm controls every input and reference before a probe is chosen.

\begin{lemma}[Lifting positive order to amplitudes]
\label{lem:amplitude-splitting}
If $N,B,Z\ge0$ and $N\le B+Z$, there exist matrices $X,F$ such that
\begin{equation}
 \sqrt N=X+F,\qquad XX^\dagger\le B,\qquad FF^\dagger\le Z.
 \label{eq:smoothing-order}
\end{equation}
\end{lemma}
\begin{proof}
Put $S=B+Z$ and use its supported inverse:
\begin{equation}
 X=BS^+\sqrt N,\qquad F=ZS^+\sqrt N.
 \label{eq:smoother}
\end{equation}
Since $\supp N\subseteq\supp S$, their sum is $\sqrt N$.
Here $S^+$ is the Moore--Penrose inverse and
$S^{+1/2}:=(S^+)^{1/2}$. On $\supp S$, both $S^{+1/2}NS^{+1/2}$ and
$S^{+1/2}YS^{+1/2}$ are positive contractions whenever $0\le Y\le S$.
Squaring the latter gives $YS^+Y\le Y$, while the former gives
$S^+NS^+\le S^+$. Hence
$XX^\dagger\le BS^+B\le B$ and
$FF^\dagger\le ZS^+Z\le Z$.
No commutation or invertibility assumption is used.
\end{proof}

Apply \Cref{lem:testing-slack,lem:amplitude-splitting} with
$B=\e^{kt}M_k$ and an optimal normalized slack
$\Tr_{B^k}Z=\Delta_k(\e^{kt})\1$. For every rate $t$, they construct
one matrix $X_k(t)$ satisfying
\begin{equation}
 X_k(t)X_k(t)^\dagger\le\e^{kt}M_k,\qquad
 \nu_k(\sqrt{N_k}-X_k(t))\le e_k(t),\qquad
 e_k(t):=\sqrt{\Delta_k(\e^{kt})}.
 \label{eq:common-approximation}
\end{equation}
The same matrix works for all probes; $e_k(t)$ is simply a testing
bound, not a separately optimized approximation profile.

\paragraph{The two scalar estimates.}
Binary relative-entropy data processing gives
$a\log(1/b)-\log2\le kD_*$ for every parallel test. A positive score
$a-\e^{kr}b$ implies $b<\e^{-kr}$, so, for $r>D_*$,
\begin{equation}
 \Delta_k(\e^{kr})\le\frac{kD_*+\log2}{kr},\qquad
 \limsup_k e_k(r)\le\sqrt{D_*/r}<1.
 \label{eq:regimes}
\end{equation}
For $v>R_*$, choose a fixed $q>1$ with $\sandD_q^\infty<v$.
Equation \eqref{eq:hockey-renyi} and superadditivity give
\begin{equation}
 \Delta_k(\e^{kv})\le
 \exp\!\left[-k(q-1)\bigl(v-\sandD_q^\infty(\cN\Vert\cM)\bigr)\right],
 \qquad e_k(v)\longrightarrow0.
 \label{eq:vanishing-regime}
\end{equation}
Neither estimate uses the endpoint or a channel strong converse.

\subsection{One tensor estimate for arbitrary entangled probes}
\label{sec:coherent}
We now translate an amplitude decomposition into a R\'enyi bound. Each
term contributes through two quantities: its domination cost relative
to $M_k$ and its amplitude norm. The estimate remains valid for probes
entangled across all repeated blocks.

Write $\|X\|_q=(\Tr|X|^q)^{1/q}$ for the Schatten $q$-norm,
where $|X|=(X^\dagger X)^{1/2}$. For a state $\beta$ and
$0\le\gamma\le\lambda\beta$,
\begin{equation}
 \sandQ_p(\gamma\Vert\beta)\le\lambda^{p-1}\Tr\gamma
 \qquad(p>1).
 \label{eq:renyi-trace-bound}
\end{equation}
For nonzero $\gamma$, apply $\sandD_p\le\Dmax$ to
$\gamma/\Tr\gamma$ and use the homogeneity
$\sandQ_p(c\gamma\Vert\beta)=c^p\sandQ_p(\gamma\Vert\beta)$.
The zero case is immediate; domination guarantees
$\supp\gamma\subseteq\supp\beta$.

\begin{lemma}[Tensor amplitude moment bound]
\label{lem:moment}
Let $\sqrt{N_k}=\sum_iX_i$ be a finite decomposition satisfying
$X_iX_i^\dagger\le\lambda_iM_k$ for $\lambda_i>0$. Then
\begin{equation}
 \exp\!\left[\frac{k(p-1)}{2p}\sandD_p^\infty(\cN\Vert\cM)\right]
 \le\sum_i\lambda_i^{(p-1)/(2p)}\nu_k(X_i)^{1/p},\qquad p>1.
 \label{eq:full-word-bound}
\end{equation}
\end{lemma}
\begin{proof}
Fix $k,p$ and repeat the decomposition $m$ times. An arbitrary probe
on all $km$ uses is represented, up to a reference isometry, by a
density operator $\sigma$;
write $L_\sigma=\sqrt\sigma\otimes\1$ and
$\beta=L_\sigma M_{km}L_\sigma$. This is a state because
$\Tr_{B^{km}}M_{km}=\1$ and $\Tr\sigma=1$.
For an amplitude $X$ with $XX^\dagger\le\lambda M_{km}$, put
$\gamma=L_\sigma XX^\dagger L_\sigma$. Then
\begin{equation}
 \gamma\le\lambda\beta,\qquad
 \operatorname{ran}(L_\sigma X)=\supp\gamma\subseteq\supp\beta.
 \label{eq:amplitude-support}
\end{equation}
Thus the negative powers below are well defined on $\supp\beta$,
including for singular $\sigma$ and singular $J^{\cM}$.
Equation \eqref{eq:renyi-trace-bound} implies
\begin{equation}
 \bigl\|\beta^{-(p-1)/(2p)}L_\sigma X\bigr\|_{2p}
 \le\lambda^{(p-1)/(2p)}
       \bigl(\Tr L_\sigma XX^\dagger L_\sigma\bigr)^{1/(2p)}
 \le\lambda^{(p-1)/(2p)}\nu_{km}(X)^{1/p}.
 \label{eq:one-amplitude-bound}
\end{equation}
For a word
$X_{\boldsymbol i}=X_{i_1}\otimes\cdots\otimes X_{i_m}$, both its order
constant and its amplitude norm factorize. The Schatten triangle
inequality therefore gives
\begin{align}
 \bigl\|\beta^{-(p-1)/(2p)}L_\sigma\sqrt{N_{km}}\bigr\|_{2p}
 &\le\sum_{i_1,\ldots,i_m}\prod_{j=1}^m
     \lambda_{i_j}^{(p-1)/(2p)}\nu_k(X_{i_j})^{1/p}\nonumber\\
 &=\left(\sum_i\lambda_i^{(p-1)/(2p)}\nu_k(X_i)^{1/p}\right)^m.
 \label{eq:word-moment}
\end{align}
The $2p$th power of the left side is the null output's sandwiched moment.
The bound is independent of $\sigma$: no product-probe assumption was
made. Optimize over $\sigma$, take logarithms, divide by $km$, and let
$m\to\infty$ at fixed $k,p$. Regularization along multiples of $k$
gives \eqref{eq:full-word-bound}.
\end{proof}

\subsection{Closing the endpoint with two approximation rates}
\label{sec:endpoint-conclusion}
\begin{proof}[Proof of \Cref{thm:endpoint}]
Suppose $D_*<R_*$. Fix $D_*<r<R_*<v$, and put
$\delta=\sqrt{D_*/r}<1$. By \eqref{eq:regimes}--\eqref{eq:vanishing-regime},
$\limsup_k e_k(r)\le\delta$ and $e_k(v)\to0$.
The common approximations give the identity
\begin{equation}
 \sqrt{N_k}
   =X_k(r)+\bigl(X_k(v)-X_k(r)\bigr)
                +\bigl(\sqrt{N_k}-X_k(v)\bigr).
 \label{eq:three-amplitudes}
\end{equation}
For these three amplitudes, admissible order constants and norm bounds are
\begin{equation}
\begin{array}{c|c|c}
 \text{amplitude}&\lambda_i\ \text{in }X_iX_i^\dagger\le\lambda_iM_k
                       &\nu_k(X_i)\ \text{at most}\\ \hline
 X_k(r)&\e^{kr}&1+e_k(r)\\
 X_k(v)-X_k(r)&4\e^{kv}&e_k(r)+e_k(v)\\
 \sqrt{N_k}-X_k(v)&4\e^{k\max\{C,v\}}&e_k(v)
\end{array}
\label{eq:amplitude-norms}
\end{equation}
The norm bounds are triangle inequalities. The order bounds follow from
$(X-Y)(X-Y)^\dagger\le2XX^\dagger+2YY^\dagger$ and
$N_k\le\e^{kC}M_k$.

Fix $s>0$. For $k>2s$, choose
\[
 p_k=\frac{k}{k-2s},\qquad
 \frac{k(p_k-1)}{2p_k}=s,\qquad \frac1{p_k}=1-\frac{2s}{k}.
\]
Applying \Cref{lem:moment} and $R_*\le\sandD_{p_k}^\infty$ gives exactly
\begin{align}
 \e^{sR_*}
 &\le\e^{sr}(1+e_k(r))^{1-2s/k}
   +4^{s/k}\e^{sv}(e_k(r)+e_k(v))^{1-2s/k}\nonumber\\
 &\quad+4^{s/k}\e^{s\max\{C,v\}}e_k(v)^{1-2s/k}.
 \label{eq:finite-block-endpoint}
\end{align}
The repetition limit in \Cref{lem:moment} was already taken at each
fixed $k,p_k$. Now let $k\to\infty$ with $r,v,s$ fixed. Since
$0\le e_k(v)\le1$ and $1-2s/k\ge1/2$ eventually, the last term vanishes.
For the first two terms, use $e_k(r)\le\delta+\zeta$ and
$e_k(v)\le\zeta$ eventually, take the upper limit, and then let
$\zeta\downarrow0$. This needs no limit for $e_k(r)$ and includes
$D_*=0$. The result is
\begin{equation}
 \e^{sR_*}\le(1+\delta)\e^{sr}+\delta\e^{sv}.
 \label{eq:profile-endpoint-bound}
\end{equation}
This holds for every fixed $v>R_*$ and $s>0$. Let $v\downarrow R_*$:
\[
 1-\delta\le(1+\delta)\e^{-s(R_*-r)}.
\]
The left side is strictly positive, whereas the right side tends to zero
as $s\to\infty$. This contradiction proves $R_*=D_*$. Monotonicity in
$p$ identifies $R_*$ with the order-one limit. The limits were taken in
the order $m\to\infty$, $k\to\infty$, $v\downarrow R_*$, $s\to\infty$;
no estimate uniform in $v$ or interchange of input optimization is used.
\end{proof}

The two approximation rates have different jobs. The rate $r$ gives an
error below one; the rate $v$ moves that nonvanishing correction to a cost
arbitrarily close to $R_*$. Only a vanishing remainder is charged at the
possibly larger cost $C$. Thus an exponentially accurate low-cost
approximation is not a prerequisite for the endpoint proof.

\subsection{Completing the channel Stein theorem}
\begin{proof}[Completion of the proof of \Cref{thm:stein}]
The lower bound and the infinite-rate case were settled at the beginning
of \cref{sec:reduction}. In the supported case, fix
$r>D^\infty(\cN\Vert\cM)$ and choose $D^\infty<t<r$.
By \Cref{thm:endpoint}, there is a fixed $q>1$ with
$\sandD_q^\infty(\cN\Vert\cM)<t$. Put
$\eta_t=(q-1)(t-\sandD_q^\infty(\cN\Vert\cM))>0$.
Equations \eqref{eq:general-score-bound} and \eqref{eq:vanishing-regime}
then give, for every general tester with $b\le\e^{-nr}$,
\begin{equation}
 a\le \e^{-n(r-t)}+g_n\e^{-n\eta_t}.
 \label{eq:direct-stein-converse}
\end{equation}
Since $\log g_n=o(n)$, there are $c>0$ and $n_0$, depending only on
the channels and $r$, such that the right side is at most $\e^{-cn}$
for $n\ge n_0$. This proves the exponential strong converse directly.
For a fixed $\varepsilon\in(0,1)$, the condition $a\ge1-\varepsilon$
therefore excludes $b\le\e^{-nr}$ for all sufficiently large $n$.
Thus the upper Stein rate is at most $r$. Let $r\downarrow D^\infty$
and combine with \eqref{eq:stein-lower-bound} and the tester inclusions.
\end{proof}

The one-shot duality and the Stein theorem give the lower AEP bound.
For the upper bound, the additional ingredient is the IID domination
\eqref{eq:slack-lift}: the endpoint theorem makes its slack exponentially
small, so its polynomial prefactor can be absorbed into any fixed budget.
\begin{corollary}[Modified smooth max-relative entropy AEP]
\label{thm:smoothing-aep}
For every finite-dimensional CPTP pair $\cN,\cM:A\to B$,
every fixed $\delta\in(0,1)$, and
$\bullet\in\{\mathrm{all},\mathrm{aff},\mathrm{prod},\mathrm{iid}\}$,
\begin{equation}
    \lim_{n\to\infty}\frac1n
    \sandD_{\max,\bullet}^{\delta}
    (\cN^{\otimes n}\Vert\cM^{\otimes n})
    =D^\infty(\cN\Vert\cM).
    \label{eq:smoothing-aep}
\end{equation}
The equality includes the extended value $+\infty$.
\end{corollary}
\begin{proof}
First assume Choi support inclusion. Fix $q\in(\delta,1)$.
Using the candidate budget $\delta$ in \eqref{eq:one-shot-par} gives
\[
 \sandD_{\max,\mathrm{all}}^\delta
 (\cN^{\otimes n}\Vert\cM^{\otimes n})
 \ge D_H^{1-q,\parr}
 (\cN^{\otimes n}\Vert\cM^{\otimes n})+\log(q-\delta).
\]
Divide by $n$ and apply \Cref{thm:stein}. By the inclusion chain
\eqref{eq:smoothing-order-chain}, this proves the lower limit bound
for all four smoothing sets.

For the upper bound, fix $t>D^\infty(\cN\Vert\cM)$.
By \Cref{thm:endpoint}, choose $p>1$ such that
$\sandD_p^\infty(\cN\Vert\cM)<t$, and set
$\eta=(p-1)(t-\sandD_p^\infty(\cN\Vert\cM))>0$.
Equations \eqref{eq:hockey-renyi} and \eqref{eq:slack-lift} give
\[
 N_n\le\e^{nt}M_n+g_n\e^{-n\eta}Q_{n,\e^{nt}},
 \qquad Q_{n,\e^{nt}}\in\mathcal K_{\mathrm{iid},n}.
\]
Since $\log g_n=o(n)$, eventually $g_n\e^{-n\eta}\le\delta$.
Increasing this coefficient to $\delta$ preserves the order inequality,
so $\sandD_{\max,\mathrm{iid}}^\delta\le nt$ for all sufficiently large
$n$. The other three entropies are no larger. Letting $t\downarrow
D^\infty$ proves the upper limit bound.

If Choi support inclusion fails, use the parallel test constructed in
\Cref{sec:reduction}, with $b_n=0$ and
$a_n=1-(1-q_0)^n\to1$. Every feasible inequality
$N_n\le\lambda M_n+\delta Q$ with $Q\in\mathcal K_{\mathrm{all},n}$
would imply
\[
 a_n\le\delta\Tr(T_{1,n}Q)\le\delta,
\]
because the parallel normalizer has unit contraction with every joint
channel Choi operator. For all sufficiently large $n$, $a_n>\delta$,
so no finite $\lambda$ is feasible. This proves eventual infinity for
all four smoothing entropies, matching $D^\infty=+\infty$.
\end{proof}

\begin{remark}[Budgets and boundary checks]
The same proof allows a sequence of budgets $\delta_n\in(0,1)$ with
$\limsup_n\delta_n<1$ and $\log(1/\delta_n)=o(n)$: choose a fixed
$q>\limsup_n\delta_n$ for the lower bound, and absorb the exponentially
small IID slack into $\delta_n$ for the upper bound. At zero budget,
however, all four entropies equal $n\Dmax(\cN\Vert\cM)$, so the
positive-budget qualification is essential.

For identical channels, direct trace comparison and $Q=N_n$ give
$\sandD_{\max,\bullet}^\delta=\log(1-\delta)$ for every class.
Likewise, $D_H^{1-q,\mathsf S}=-\log q$. These formulas check the
normalization and the sign of the logarithmic correction in the duality.
\end{remark}

\section{Finite-use bounds and the exact strong-converse exponent}
\label{sec:finite}

The preceding proof establishes exponential decay above the Stein rate.
We now retain the optimal scalar constant in the testing bound to
quantify that decay. The resulting finite-use inequality, together with
known parallel achievability, identifies the exact exponent for all
three tester classes.

\subsection{The full fixed-order converse}
The right side of \eqref{eq:hockey-renyi} can be multiplied by the
sharper constant $c_p=(p-1)^{p-1}/p^p$. Indeed, maximizing
$(1-x)x^{p-1}$ for $x=\lambda b/a\in(0,1)$ gives
\[
 a-\lambda b\le c_p\lambda^{1-p}a^pb^{1-p}.
\]
Insert this strengthened bound on $\Delta_n$ into
\eqref{eq:general-score-bound} and minimize over $\lambda>0$.
The elementary identity
$\inf_{\lambda>0}\{\lambda b+c_pK\lambda^{1-p}\}
=K^{1/p}b^{(p-1)/p}$ gives
\begin{equation}
 \frac{a^pb^{1-p}}{g_n}
 \le \exp\bigl[(p-1)\sandD_p(\cN^{\otimes n}\Vert\cM^{\otimes n})\bigr]
 \le \e^{n(p-1)\sandD_p^\infty(\cN\Vert\cM)}.
 \label{eq:binary-converse}
\end{equation}
At $(a,b)=(0,0)$, the binary moment is defined to be zero.
Here $b=0$ implies $a=0$ by finite Choi domination; otherwise the
minimization applies directly. In particular, $a\ge1-\varepsilon$ implies
\begin{equation}
 D_H^{\varepsilon,\gen}(\cN^{\otimes n}\Vert\cM^{\otimes n})
 \le n\sandD_p^\infty(\cN\Vert\cM)
       +\frac{\log g_n+p\log(1/(1-\varepsilon))}{p-1}.
 \label{eq:finite-converse}
\end{equation}

For $b\le\e^{-nr}$, \eqref{eq:binary-converse} also gives
\begin{equation}
 a\le g_n^{1/p}
 \exp\!\left[-n\frac{p-1}{p}
       \bigl(r-\sandD_p^\infty(\cN\Vert\cM)\bigr)\right].
 \label{eq:exponential-converse}
\end{equation}
These bounds sharpen \eqref{eq:direct-stein-converse} and will identify
the exact exponent. Their proof uses only the scalar inequality
\eqref{eq:general-score-bound}; no explicit parallel tester is needed.

\subsection{The exact strong-converse exponent}
\label{sec:exact-exponent}
The preceding proof establishes the Stein theorem without using an
external exact-exponent formula. To identify the entire decay exponent,
we now use the parallel achievability part of Fawzi and Fawzi's
Theorem~5.5 \cite{FawziFawzi2021}. For a tester class $\mathsf S$ and
$r>0$, define
\begin{equation}
 E_{\rm sc}^{\mathsf S}(r)=
 \inf_{\substack{(T_{1,n},T_{2,n})\in\mathsf S\ \text{for all }n\\
        \liminf_n-n^{-1}\log b_n\ge r}}
       \limsup_{n\to\infty}-\frac1n\log a_n.
 \label{eq:sc-definition}
\end{equation}
Here $-\log0=+\infty$. In \cite[Section 5.2.1]{FawziFawzi2021},
the errors are $\alpha_n=1-a_n$ and $\beta_n=b_n$; their condition
$\limsup_n n^{-1}\log\beta_n\le-r$ is exactly the condition above.
Their base-two convention is converted to natural logarithms by
replacing $r$ with $r/\log2$ and multiplying divergences and exponents
by $\log2$. The cited theorem is stated for finite-dimensional CPTP
maps and explicitly achieves the exponent with nonadaptive, hence
parallel, strategies. Our corollary assumes Choi support inclusion;
neither channel is required to have a full-rank Choi operator.

\begin{corollary}[Exact exponent and its threshold]
\label{cor:exact-exponent}
For finite-dimensional CPTP maps $\cN,\cM:A\to B$ with
$\supp J^{\cN}\subseteq\supp J^{\cM}$, every $r>0$, and
$\mathsf S\in\{\parr,\ada,\gen\}$,
\begin{equation}
E_{\rm sc}^{\mathsf S}(r)
=\sup_{p>1}\frac{p-1}{p}
\bigl[r-\sandD_p^\infty(\cN\Vert\cM)\bigr].
\label{eq:exact-exponent}
\end{equation}
The common exponent is strictly positive if and only if
$r>D^\infty(\cN\Vert\cM)$.
\end{corollary}

\begin{proof}
For every sequence in the definition and every fixed $p>1$,
\eqref{eq:binary-converse} gives
\[
 \limsup_n-\frac1n\log a_n\ge
       \frac{p-1}{p}\bigl[r-\sandD_p^\infty(\cN\Vert\cM)\bigr].
\]
Indeed, first replace $r$ by $r-\eta$ in the eventual bound on $b_n$,
then let $\eta\downarrow0$; $n^{-1}\log g_n\to0$. Taking the supremum
over $p$ gives a lower bound on the general exponent. Parallel achievability in
\cite[Theorem 5.5]{FawziFawzi2021} bounds the parallel exponent from
above by the same expression. The tester inclusions give
$E_{\rm sc}^{\gen}\le E_{\rm sc}^{\ada}\le E_{\rm sc}^{\parr}$,
so all three values coincide. This equality does not use the endpoint;
the endpoint identifies its positivity threshold.
By \Cref{thm:endpoint}, the supremum is positive for $r>D^\infty$.
For $r\le D^\infty$, its terms are nonpositive and approach zero as
$p\downarrow1$, since $\sandD_p^\infty\le\Dmax<\infty$.
\end{proof}

\section{Discussion}\label{sec:discussion}
For finite-dimensional memoryless channels, feedback and indefinite
causal order do not change the fixed-error Stein rate. At every finite
Type-II exponent above $D^\infty(\cN\Vert\cM)$, correct null
acceptance decays exponentially. When
$\supp J^{\cN}\subseteq\supp J^{\cM}$, the exact parallel
strong-converse exponent at every $r>0$ also applies to adaptive and
general testers.
These conclusions identify the asymptotic testing threshold even
though the tester classes may perform differently at finite blocklength.

A single normalized positive slack organizes the proof. Omitting the
redundant positivity constraint on the tester's marginal makes its SDP
dual satisfy an exact partial-trace normalization. Fixed-marginal de
Finetti reduction then gives \eqref{eq:slack-lift}, which transfers
parallel testing bounds to every general tester and supplies the IID
smoother in the AEP. Independently, lifting that slack to amplitudes
controls every input before tensorization. The three-amplitude comparison
identifies the regularized R\'enyi endpoint without exchanging optimization
and asymptotic limits.

The one-shot duality further identifies which smoothing constraint
belongs to each tester normalization. Alongside the all-channel and
positive-affine sets in the exact dualities, the convex product and
IID sets have the same asymptotic limit: all four modified smooth
entropies converge to
$D^\infty(\cN\Vert\cM)$ at every fixed budget in $(0,1)$.
The proof uses the endpoint to make the Choi slack exponentially
small, then absorbs the polynomial cost of its IID domination.

The explicit conversion in \Cref{app:parallelization} scales both
acceptance probabilities by $1/g_n$; it
does not imply equality of finite-use tests. A second-order analysis
would require quantitative control of the common-amplitude error as
the rate approaches $D^\infty$ and the R\'enyi order approaches one.
Such a bound would have to track the $O(\log n)$ postselection term in the finite-copy converse \eqref{eq:finite-converse}. The present
endpoint proof establishes continuity but supplies no such rate of
convergence.

\paragraph{Note added.} 
After completing the initial version of this note on September 7,
2026~\cite{private_email}, we learned of independent work by Gao, Ji, and Liu
\cite{GaoJiLiu2026}, posted on September 23, 2026, on the fixed-error
channel Stein theorem and exponential strong converse for parallel
and adaptive strategies. The argument in this manuscript shares the testing-to-amplitude construction with Gao, Ji, and Liu~\cite{GaoJiLiu2026}, but closes the endpoint directly through this three-amplitude bound; their proof first makes the approximation exponentially accurate through iterative rate reduction and a
fixed-block tensor expansion.

\section*{Acknowledgements}
The authors thank Ludovico Lami, Bartosz Regula, Yinan Li, and Li Gao for insightful discussions. During his visit to HKUST(GZ), Bartosz Regula raised the question of whether indefinite causal order could improve channel discrimination; we are especially grateful to him for this suggestion.
This work was supported in part by the National Key R\&D Program
of China (Grant No.~2024YFE0102500) and the National Natural Science
Foundation of China (Grant Nos.~92576114, 12447107).
Additional support was provided by the Guangdong Provincial Quantum
Science Strategic Initiative (Grant Nos.~GDZX2403008, GDZX2503001)
and the Guangdong Provincial
Key Lab of Integrated Communication, Sensing and Computation for
Ubiquitous Internet of Things (Grant No.~2023B1212010007).

\Needspace{14\baselineskip}
\paragraph{AI disclosure.}
The authors formulated the original problem and obtained the results leading up to~\Cref{thm:endpoint} without AI assistance, including the one-shot duality, polynomial reduction for general testers and the fixed-order R\'enyi converse. GPT-5.6 sol then generated the initial version of~\Cref{thm:endpoint}. The authors verified and refined the generated proof to the best of our knowledge and take full responsibility for the proof and claim.

\appendix
\Needspace{15\baselineskip}
\section{Exact polynomial parallelization}
\label{app:parallelization}
The main proof only uses the scalar consequence \eqref{eq:general-score-bound}.
For completeness, we retain the stronger exact acceptance-pair conversion.

\begin{proposition}[Polynomial parallelization]
\label{prop:reduction}
Every $n$-slot general tester with acceptance pair $(a,b)$ admits an
$n$-slot parallel tester with pair $(a/g_n,b/g_n)$. Consequently,
\begin{equation}
 \beta_\varepsilon^{\gen}(\cN^{\otimes n},\cM^{\otimes n})
 \ge g_n\,
 \beta_{1-(1-\varepsilon)/g_n}^{\parr}
       (\cN^{\otimes n},\cM^{\otimes n}).
 \label{eq:roc-reduction}
\end{equation}
\end{proposition}
\begin{proof}
Jointly twirl the tester without changing $(a,b)$, and put $W=T_1+T_2$.
Since $X\otimes\1\ge W\ge0$ already implies $X\ge0$, SDP duality gives
\begin{equation}
 \gamma(W)=\min_{\substack{X=X^\dagger\\X\otimes\1\ge W}}\Tr X
          =\max_{\substack{Q\ge0\\\Tr_{B^n}Q=\1_{A^n}}}\Tr WQ.
 \label{eq:parallel-normalizer-cost}
\end{equation}
Strict primal feasibility follows from a large scalar $X$; the minimum
is attained by trace compactness, and the dual maximum is attained
because the channel Choi set is compact. Twirl a dual optimizer $Q$ and apply
\eqref{eq:definetti}. General normalization then gives
\[
 \gamma(W)\le g_n\int\Tr\bigl[W(J^{\cC})^{\otimes n}\bigr]
                     \,\mathrm d\mu(\cC)=g_n.
\]
Moreover, $\gamma(W)\ge\Tr(WP_0)=1$, where $P_0$ is the product
depolarizing Choi operator. Thus the following normalization never
divides by zero. For an optimal $X$, put $\sigma=X/\Tr X$. Then
$W\le g_n\sigma\otimes\1$, so
$T_1/g_n$ and $\sigma\otimes\1-T_1/g_n$ form the required parallel
tester. Their acceptance pair is exactly $(a/g_n,b/g_n)$.
\end{proof}

\end{document}